\documentclass[11pt,letterpaper]{article}

\usepackage[margin=1in]{geometry}
\usepackage{amsmath,amssymb,amsthm,mathtools}
\usepackage{booktabs}
\usepackage{enumitem}
\usepackage{microtype}
\usepackage[hidelinks]{hyperref}

\newtheorem{theorem}{Theorem}[section]
\newtheorem{lemma}[theorem]{Lemma}
\newtheorem{proposition}[theorem]{Proposition}
\newtheorem{corollary}[theorem]{Corollary}
\theoremstyle{definition}
\newtheorem{definition}[theorem]{Definition}
\newtheorem{example}[theorem]{Example}
\theoremstyle{remark}
\newtheorem{remark}[theorem]{Remark}

\newcommand{\EFXz}{\mathrm{EFX}_0}
\newcommand{\EFXp}{\mathrm{EFX}^{+}}
\newcommand{\RI}{\mathrm{RI}}

\newcommand{\poly}{\operatorname{poly}}

\hypersetup{
  pdftitle={Rival-Injective Allocations: Support-List Structure and Maximum-Anchor EFX0 Certificates},
  pdfauthor={Junshuo Wang}
}

\title{Rival-Injective Allocations: Support-List Structure and
Maximum-Anchor \(\EFXz\) Certificates}
\author{Junshuo Wang\thanks{This work was developed with substantial assistance
from OpenAI ChatGPT and Codex agent workflows.  The author directed the research
program and retains final authority and responsibility for the manuscript and
its release; see the detailed
contribution and responsibility statement at the end.}\\[0.4em]
\small University of Electronic Science and Technology of China\\
\small \texttt{202522220534@std.uestc.edu.cn}}
\date{August 29, 2026}

\begin{document}
\maketitle

\begin{abstract}
We study complete allocations under nonnegative additive valuations, focusing
on the all-good form of envy-freeness up to any good (\(\EFXz\)), through the
positive supports of the goods.  We define
\emph{rival-injective} (RI) endpoint ownership: each good with nonempty positive
support is assigned to an agent who values it positively, and every ordered
observer--owner pair is used by at most one good.  RI ownership is exactly
proper list coloring of the graph
that joins two goods when their positive supports overlap in at least two
agents.  For pair-supported goods together with arbitrary exceptional goods,
fixing the exceptional owners yields a necessary-and-sufficient pair-capacity
completion criterion and an exact finite-domain owner constraint satisfaction
problem (CSP).

To connect this support structure to all-good fairness, we use the elementary
implication that every RI allocation in which each agent's own bundle is worth
at least her maximum-valued singleton is all-good \(\EFXz\).  We call the
own-bundle condition \emph{maximum-singleton dominance}.  A specified injective choice of
maximum-singleton anchors, together with residual support-list degeneracy,
constructs an RI allocation with this dominance by reverse greedy coloring.
Given the anchors, verification and construction take \(O(nm^2)\) time.  We
give two explicit witness families for this sufficient certificate: a
nonempty relatively open, 22-dimensional cone on a fixed
\(4\times10\) support face and a family for every \(n\ge4\) with two
universal-support goods and
\(m=(n-1)(n-2)+2\) goods.  Both families fail unanchored list degeneracy and
are not implied by the explicit pure-multigraph or published
high-girth/controlled-multiplicity hypotheses compared here.  We further study
recognition of the maximum-anchor certificate class, while leaving its general
complexity unresolved.  The exact list-coloring and pair-capacity
characterizations concern RI ownership, not general all-good \(\EFXz\)
existence; unrestricted four-agent, ten-good all-good \(\EFXz\) remains
unresolved.
\end{abstract}

\section{Introduction}

The existence of complete allocations that are envy-free up to any good
(all-good \(\EFXz\)) is a central problem in fair division of indivisible goods.
Exact existence is known for three additive agents, for four additive agents
with at most nine goods, and for several restricted valuation domains, but
remains open for general additive valuations with four or more agents
\cite{ChaudhuryGargMehlhorn2024,AlkassarFouzMehlhorn2026,
PlautRoughgarden2020,ByrkaEtAl2026}.
Recent counterexamples for broader monotone valuations can already be made
submodular, so the unresolved existence question invoked here is specifically
the additive one \cite{AkramiEtAlCounterexample2026}.
Support-restricted models provide a complementary route: graphical,
multigraph, and
selected hypergraph instances admit exact EFX theorems under several
structural hypotheses
\cite{ChristodoulouEtAl2023,BhaskarPandit2025,AfshinmehrEtAl2026,LianeasEtAl2026}.
These results make support overlap a natural intermediate object between
unrestricted additive valuations and purely graph-supported instances.  They
suggest a concrete question:
\begin{quote}
Which local restrictions on support overlap still yield transparent, complete
EFX allocations when the support hypergraph contains short cycles and genuine
high-support goods?
\end{quote}

We address the question for nonnegative additive valuations by following the
positive support of each good.  Assigning a good to an agent in its support can
be viewed as choosing an endpoint of its support hyperedge.  The relevant local
resource is an ordered pair \((i,j)\): observer \(i\) sees a good owned by
agent \(j\).  If no ordered resource is used twice, then each observer sees at
most one positively valued good in any rival bundle, and deleting that good
eliminates the rival's entire visible value.  We call such endpoint ownership
\emph{rival-injective} (RI).  This is the canonical support-only one-sparsity
condition suggested by the one-good deletion quantifier.  Writing
\(S_g=\{k:v_k(g)>0\}\) in this overview, RI equivalently says that, for a fixed
owner \(j\), the traces \(S_g\setminus\{j\}\) of the goods assigned to \(j\)
are pairwise disjoint.

The distinction between positive-good EFX (\(\EFXp\)) and all-good EFX
(\(\EFXz\)) is visible in the
two-good rival bundle \(\{p,z\}\), where observer \(i\) has \(v_i(p)>0\) and
\(v_i(z)=0\).  RI makes deletion of \(p\) harmless, but deleting \(z\) leaves
the visible singleton \(p\).  If \(i\) owns a maximum-singleton good \(h_i\),
then \(v_i(h_i)\ge v_i(p)\); this is the dominance supplied by our anchors.

The same condition has an exact coloring interpretation.  Form a conflict
graph on goods, joining two goods precisely when their positive supports share
at least two agents, and give each good its support as a color list.  Choosing
the owner is then choosing a list color, and RI is equivalent to proper list
coloring.  Thus support overlap leads first to an exact structural translation,
not directly to an exact characterization of either \(\EFXp\) or \(\EFXz\).  In
the important case of
pair-supported goods plus exceptional high-support goods, the pair goods can be
eliminated after the exceptional owners are fixed.  Each unordered agent pair
contributes two directed observer--owner capacities, producing an exact owner
CSP; two surviving candidates per exceptional good give the familiar
2-satisfiability (2-SAT) special case.

RI by itself controls only goods that an observer values positively.  The
missing ingredient for all-good \(\EFXz\) is value dominance: an observer
must also tolerate removal of a good she values at zero while a visible good
remains in the rival bundle.  It is enough that every agent values her own
bundle at least as much as her maximum-valued singleton; maximum anchors are
one constructive way to supply this dominance.  We precolor one distinct
maximum-singleton good for each active agent and delete from each residual
support list the colors blocked by adjacent anchors.  Residual list degeneracy
then lets reverse greedy coloring extend the precoloring.

\paragraph{Logical roadmap.}
The support sets \(S_g\), singleton maxima \(M_i\), and conflict graph \(C_2\)
are derived from the valuation instance; they are not extra hypotheses.  An
owner map and an anchor map are certificate choices.  Anchor residual list
degeneracy (A-LD) is a sufficient
condition on a specified anchor choice, not a property asserted for every
instance.  The proof dependencies are
\[
v\longmapsto (S_g,M_i,C_2)
\]
and
\[
\text{specified anchors + A-LD}
\Longrightarrow \text{RI + maximum-singleton dominance}
\Longrightarrow \EFXz .
\]
The first arrow records definitions; the second is the coloring construction;
the third is the elementary fairness implication used to organize the proof.

\paragraph{Contributions.}
The elementary RI--dominance implication organizes the construction and is not
an independent novelty claim; related sparse-bundle reasoning appears in prior
EFX and envy-freeness up to one less-preferred good (EFL) proofs discussed in
the Related Work section.  Our contribution is to make the support-side
mechanism explicit, exact, and algorithmic.
\begin{enumerate}[leftmargin=2em]
\item \emph{Structural representation.}  We prove that RI endpoint ownership
is equivalent to support-list coloring of the two-overlap conflict graph, so
support overlap becomes an exact coloring problem.
\item \emph{Exact completion structure.}  For pair goods and exceptional
goods, fixed exceptional owners admit an RI completion exactly when three
directed pair-capacity inequalities hold; this yields an exact
exceptional-owner CSP.  The two-choice 2-SAT case is a standard algorithmic
corollary, not a separate novelty claim.
\item \emph{Certificate realization.}  Specified injective
maximum-singleton anchors and residual list degeneracy construct RI together
with the dominance needed for complete \(\EFXz\).  Given the anchors, the
certificate is verified and an allocation is constructed in \(O(nm^2)\) time.
\item \emph{Recognition.}  A fixed anchor map is
accepted exactly when its generalized residual core is empty.  This gives
polynomial recognition for fixed \(n\), slice-wise polynomial (XP) parameterized
by the number of active agents, and fixed-parameter tractable (FPT) parameterized
by the total maximum-tie surplus
\(p=\sum_{i\in N^+}(|T_i|-1)\), an exact
two-active-agent characterization, and a polynomial safe-anchor sufficient
subclass.  We do not resolve the general recognition problem.
\item \emph{Witness families.}  We give a nonempty
relatively open \(\EFXz\) cone of full dimension 22 within a fixed
\(4\times10\) support face, and a family \(\mathcal U_n\) for every \(n\ge4\)
with two universal goods and saturated ordered resources.  These families fail
unanchored list degeneracy and the explicit pure-multigraph and
high-girth/controlled-multiplicity hypotheses examined in this paper.
\end{enumerate}

\paragraph{Positioning.}
Favorite-good arguments and capacitated orientations have substantial prior
literatures \cite{MarkakisSantorinaios2023,LohPagh2014}.  Reverse-greedy list
coloring and the elementary two-choice Boolean encoding are generic tools.  The
elementary RI--dominance implication is used only to organize the construction.
The claims specific to this paper concern the support-derived RI translation,
its exact ordered-resource completion system, the maximum-anchor residual-list
machinery, recognition, and the witness constructions, not the underlying
generic tools.
Likewise, failure of a cited sufficient hypothesis shows that a theorem does
not directly imply our family; it does not show that all methods in the prior
literature fail on every valuation in that family.

\section{Model and two EFX notions}

Let \(N\) be a finite nonempty set of \(n\) agents and let \(G\) be a finite set
of \(m\) indivisible goods.  Agent \(i\) has a nonnegative additive valuation
\(v_i\), so \(v_i(B)=\sum_{g\in B}v_i(g)\).  For bit-complexity statements,
valuation entries are nonnegative rationals encoded in binary; equivalent
support/maximum-set combinatorial input is specified in
Section~\ref{sec:recognition}.  An allocation
\(A=(A_i)_{i\in N}\) is \emph{complete} if the bundles partition \(G\).

\begin{definition}[Positive-good EFX and all-good EFX]
An allocation is \(\EFXp\) if, for all \(i,j\in N\) and every
\(g\in A_j\) with \(v_i(g)>0\),
\[
v_i(A_i)\ge v_i(A_j\setminus\{g\}).
\]
It is \(\EFXz\) if the same inequality holds for every \(g\in A_j\),
including goods of value zero to observer \(i\).
\end{definition}

Thus \(\EFXz\) is the stronger notion on instances containing zero values.  We
write
\[
S_g=\{i\in N:v_i(g)>0\},
\]
for the positive support of good \(g\).  A good with \(S_g=\varnothing\) is
called a zero-support good.  Such goods can be restored after all
positive-support goods have been allocated.  Write
\[
G^+=\{g\in G:S_g\ne\varnothing\}.
\]

Terminology is not uniform across the literature.  Some papers use ``EFX''
for the all-good definition above, whereas others reserve ``EFX0'' for that
definition and use ``EFX'' only for observer-positive removed goods.  We use
\(\EFXp\) and \(\EFXz\) in our formal definitions and theorem statements to keep
the two quantifier domains explicit, and we translate external results according
to their displayed definitions rather than their labels
\cite{PlautRoughgarden2020,AmanatidisEtAl2021}.

\begin{definition}[Zero-endpoint padding for theorem comparison]
\label{def:zero-padding}
For a good \(g\) with positive support \(S_g\), a \emph{zero-endpoint padding}
is a hyperedge representation
\[
\widehat S_g\supseteq S_g,
\]
in which every added endpoint \(i\in\widehat S_g\setminus S_g\) satisfies
\(v_i(g)=0\).  A collection \((\widehat S_g)_{g\in G^+}\) is admissible for a
cited hypergraph theorem only when that theorem's model permits zero-valuing
endpoints and all of its other representation conventions are met.
\end{definition}

Positive support is the minimal such representation, but it need not be the
only representation allowed by an external model.  Consequently, to show that
a valuation instance does not satisfy a published support predicate, we must
exclude every admissible padding.  Even then the conclusion is only that the
displayed sufficient hypothesis does not directly imply the instance; it does
not rule out another algorithm or argument from the same paper.

\section{Rival injectivity and support-list coloring}

\begin{definition}[Endpoint ownership and rival injectivity]
An endpoint owner map assigns every positive-support good \(g\) an owner
\(o(g)\in S_g\).  It is \emph{rival-injective}, abbreviated \(\RI\), if for
every ordered pair \(i\ne j\),
\[
\left|\{g\in G^+:i\in S_g,\ o(g)=j\}\right|\le1.
\]
\end{definition}

For each fixed owner \(j\), RI is equivalently the pairwise disjointness of
the support traces \(S_g\setminus\{j\}\) over goods with \(o(g)=j\).  Thus a
rival bundle is one-sparse in every outside observer's positive-support row.

\begin{lemma}[RI implies positive-good EFX]\label{lem:ri-prompt}
Every rival-injective endpoint owner map extends to a complete \(\EFXp\)
allocation.
\end{lemma}
\begin{proof}
Give each positive-support good to its owner and restore zero-support goods
arbitrarily.  Fix \(i\ne j\).  Rival injectivity says that \(A_j\) contains at
most one good valued positively by \(i\).  If \(g\in A_j\) is positive to \(i\),
then deleting \(g\) leaves value zero to \(i\), which is at most the
nonnegative value of \(A_i\).
Self-comparisons follow from nonnegativity.
\end{proof}

Define the \emph{two-overlap conflict graph} \(C_2\) on the positive-support
goods by
\[
gh\in E(C_2)\quad\Longleftrightarrow\quad |S_g\cap S_h|\ge2,
\]
and give vertex \(g\) the color list \(L(g)=S_g\).

\begin{theorem}[RI--list-coloring equivalence]\label{thm:ri-list}
An endpoint owner map is rival-injective if and only if it is a proper list
coloring of \(C_2\) from the lists \(S_g\).
\end{theorem}
\begin{proof}
If adjacent goods \(g,h\) receive the same owner \(j\), their support
intersection contains \(j\) and at least one other agent \(i\).  Observer \(i\)
then sees two goods owned by rival \(j\), violating RI.

Conversely, if RI fails for \((i,j)\), two distinct goods \(g,h\) are both
visible to \(i\) and both owned by \(j\).  Endpoint ownership puts \(j\) in both
supports, so the intersection contains the distinct agents \(i,j\).  Hence
\(g,h\) are adjacent and have the same color.
\end{proof}

\begin{definition}[Support-list degeneracy]
The support instance is list-degenerate if every nonempty
\(Y\subseteq V(C_2)\) contains a good \(g\) such that
\[
d_{C_2[Y]}(g)<|S_g|.
\]
\end{definition}

\begin{corollary}\label{cor:ld}
Support-list degeneracy is a polynomially recognizable sufficient condition for
a complete \(\EFXp\) allocation.
\end{corollary}
\begin{proof}
Repeatedly delete a vertex satisfying the strict degree bound, and color in
reverse deletion order.  At most \(|S_g|-1\) colors are forbidden when \(g\) is
inserted.  Apply Theorem~\ref{thm:ri-list} and Lemma~\ref{lem:ri-prompt}.
\end{proof}

\section{Pair capacities and exceptional hyperedges}

Partition the positive-support goods into pair goods
\[
P=\{g:|S_g|=2\}
\]
and exceptional goods
\[
X=\{g:|S_g|\ge3\}.
\]
Support-one goods are forced to their unique endpoint and consume no rival
resource.  For an unordered pair \(\{i,j\}\), let
\[
\mu_{ij}=|\{g\in P:S_g=\{i,j\}\}|.
\]
Fix an owner choice \(\phi(x)\in S_x\) for each \(x\in X\), and set
\[
L_{ij}(\phi)=
|\{x\in X:\phi(x)=j,\ i\in S_x\setminus\{j\}\}|.
\]

Once \(\phi\) is fixed, the remaining calculation is an ordinary pairwise
orientation-capacity elimination: a \(\{i,j\}\)-good chooses one of two
directions, each with residual capacity at most one.  The point of the next
theorem is the exact reduction from RI completion to these independent
directed resources, not a new generic orientation technique.  This is analogous
only at a broad level to capacitated-orientability models; Loh and Pagh study a
different random-hypergraph setting \cite{LohPagh2014}.

\begin{theorem}[Exact pair completion]\label{thm:pair-capacity}
For fixed exceptional owners \(\phi\), the pair goods admit endpoint
orientations completing \(\phi\) to an RI owner map if and only if, for every
unordered pair \(\{i,j\}\),
\[
L_{ij}\le1,\qquad L_{ji}\le1,\qquad
\mu_{ij}+L_{ij}+L_{ji}\le2.
\]
\end{theorem}
\begin{proof}
Suppose \(k\) of the \(\mu_{ij}\) pair goods are assigned to \(j\), and the
remaining \(\mu_{ij}-k\) to \(i\).  The two ordered capacities are
\[
L_{ij}+k\le1,\qquad
L_{ji}+\mu_{ij}-k\le1.
\]
An integer \(k\) exists exactly when
\[
\max\{0,\mu_{ij}+L_{ji}-1\}
\le
\min\{\mu_{ij},1-L_{ij}\},
\]
which is equivalent to the three displayed inequalities.  Different unordered
pairs use disjoint ordered resources, so their completions are independent.
\end{proof}

If some \(\mu_{ij}>2\), RI is impossible.  Otherwise define the surviving
candidates
\[
C_x=\{j\in S_x:\mu_{ij}\le1\text{ for all }i\in S_x\setminus\{j\}\}.
\]
For \(i\ne j\), let
\[
Q_{ij}=\{(x,j):j\in C_x,\ i\in S_x\setminus\{j\}\}.
\]

\begin{proposition}[Exact exceptional-owner CSP]\label{prop:csp}
Assume \(\mu_{ij}\le2\) for every unordered pair \(\{i,j\}\).
An RI owner map exists if and only if one candidate is selected from each
\(C_x\), at most one selected candidate belongs to each \(Q_{ij}\), and, when
\(\mu_{ij}=1\), at most one selected candidate belongs to
\(Q_{ij}\cup Q_{ji}\).
\end{proposition}
\begin{proof}
The unary deletion of a candidate touching a pair of multiplicity two is
necessary.  The remaining two families of at-most-one constraints are exactly
\(L_{ij}\le1\) and \(L_{ij}+L_{ji}\le1\) when \(\mu_{ij}=1\).
Theorem~\ref{thm:pair-capacity} proves both directions.
\end{proof}

\begin{corollary}[Two-choice 2-SAT]\label{cor:2sat}
If every \(C_x\) is nonempty and has size at most two, RI feasibility is
decidable in polynomial time by 2-SAT.
\end{corollary}
\begin{proof}
Use an exactly-one clause for the one or two candidates of each exceptional
good.  For every two distinct candidate identities occurring in a forbidden
set \(Q_{ij}\), or in \(Q_{ij}\cup Q_{ji}\) with \(\mu_{ij}=1\), add a negative
binary clause.  Proposition~\ref{prop:csp} gives exactness.  This is the
elementary Boolean encoding of a two-element list-conflict system.
\end{proof}

\begin{corollary}[Fixed number of wide exceptions]\label{cor:wide}
Let \(t=|\{x\in X:|C_x|>2\}|\).  RI feasibility can be decided using at most
\[
\prod_{x:|C_x|>2}|C_x|\le n^t
\]
calls to 2-SAT.  Consequently, the problem is polynomial for every fixed
\(t\), but this bound is XP rather than FPT in \(t\).
\end{corollary}
\begin{proof}
First reject if some \(C_x\) is empty.  Enumerate an owner for every wide
exception.  Reject a branch if its fixed choices already violate an
ordered-resource constraint.  For each remaining narrow candidate that
conflicts with a fixed choice, add the negative unit clause forbidding that
candidate; if a good loses all candidates, reject the branch.  The remaining
exactly-one and pair-conflict constraints are conjunctive normal form with at
most two literals per clause (2-CNF), as in Corollary~\ref{cor:2sat}.  The
number of branches is the displayed product.
\end{proof}

\section{An RI--dominance bridge and maximum-anchor certificates}

Set
\[
M_i=\max_{g\in G}v_i(g),\qquad
N^+=\{i\in N:M_i>0\},
\]
with \(M_i=0\) when \(G=\varnothing\).

\begin{definition}[Maximum-singleton dominance]
\label{def:maximum-dominance}
A complete allocation \(A\) has \emph{maximum-singleton dominance} if
\[
v_i(A_i)\ge M_i
\qquad\text{for every }i\in N.
\]
This is a property of the chosen allocation, not an assumption that every
allocation of the valuation instance satisfies.
\end{definition}

The following elementary implication isolates the fairness step used by the
certificate construction.

\begin{proposition}[RI--dominance bridge]
\label{prop:ri-dominance}
Let \(o\) be an RI endpoint owner map.  Give every positive-support good \(g\)
to \(o(g)\), and restore the zero-support goods arbitrarily.  If the resulting
complete allocation has maximum-singleton dominance, then it is \(\EFXz\).
\end{proposition}
\begin{proof}
Let \(A\) be the resulting allocation and fix \(i,j\in N\) and \(g\in A_j\).
If \(i=j\), nonnegativity gives
\(v_i(A_i)\ge v_i(A_i\setminus\{g\})\).

Suppose \(i\ne j\).  By RI, \(A_j\) contains at most one positive-support good
\(p\) with \(v_i(p)>0\).  If no such \(p\) exists, then
\(v_i(A_j\setminus\{g\})=0\).  If \(g=p\), deleting \(g\) again leaves value
zero.  Otherwise every good remaining in \(A_j\) is worth zero to \(i\), except
possibly \(p\), and hence
\[
v_i(A_j\setminus\{g\})=v_i(p)\le M_i\le v_i(A_i).
\]
This covers observer-zero removed goods as well as positive removed goods.
When \(i\) is inactive, \(M_i=0\) and her valuation row is identically zero,
so the same argument is automatic.  A zero-support good is worth zero to every
observer and therefore neither creates a second visible good nor changes any
of these inequalities.  Thus arbitrary restoration is harmless, and every
all-good comparison holds.
\end{proof}

We now give one polynomially checkable way to construct the two properties in
Proposition~\ref{prop:ri-dominance}.  Suppose a specified injection
\(i\mapsto h_i\) from \(N^+\) to the positive-support goods satisfies
\[
i\in S_{h_i},\qquad v_i(h_i)=M_i.
\]
The injection is a certificate choice: its existence is not built into the
valuation model.  Let \(H=\{h_i:i\in N^+\}\), precolor \(h_i\) by \(i\), and
define the residual list
\[
L_H(g)=S_g\setminus
\{i\in N^+:gh_i\in E(C_2)\},
\qquad g\notin H.
\]

\begin{definition}[Anchor residual list degeneracy]
We say that the anchor map satisfies \emph{anchor residual list degeneracy}
(A-LD) if every nonempty
\(Y\subseteq V(C_2)\setminus H\) contains \(g\in Y\) such that
\[
d_{C_2[Y]}(g)<|L_H(g)|.
\]
\end{definition}

For fixed supports and a specified anchor map, \(C_2\) and every residual list
\(L_H(g)\) are uniquely determined.  The definition quantifies over all
nonempty residual vertex sets, so A-LD does not depend on choosing a peeling
order in advance.  Any valid first deletion leaves a smaller set to which the
same quantified condition applies.

\begin{theorem}[Maximum-anchor A-LD certificate]\label{thm:max-anchor}
If a specified injective maximum-anchor map satisfies A-LD, then there is a
complete allocation that is RI, has maximum-singleton dominance, and therefore
is \(\EFXz\).  Given the anchors, the certificate can be verified and an
allocation constructed in \(O(nm^2)\) time.
\end{theorem}
\begin{proof}
The anchors form a proper partial list coloring: their goods and colors are
both distinct.  Peel the residual graph using A-LD and color it in reverse
order from \(L_H(g)\).  The residual degree bound leaves a color, while the
definition of \(L_H(g)\) prevents conflicts with precolored anchors.  The
result is a proper support-list coloring, hence an RI owner map by
Theorem~\ref{thm:ri-list}.

Every active agent owns a maximum-valued anchor good, so nonnegativity gives
\[
v_i(A_i)\ge v_i(h_i)=M_i.
\]
An inactive agent has an all-zero valuation row, so the same inequality holds.
Thus the allocation has maximum-singleton dominance.  Restore zero-support
goods arbitrarily and apply Proposition~\ref{prop:ri-dominance}.

Constructing \(C_2\) naively takes \(O(nm^2)\); list construction, peeling, and
reverse coloring fit within this bound.
\end{proof}

\begin{remark}[Strictness and quantifiers]
The A-LD inequality must be strict.  The theorem asserts that a particular
maximum-anchor injection and its associated residual lists form a certificate.
It does not assert that every maximum matching works, or that a successful
anchor map can always be found in polynomial time.  RI is sufficient but not
necessary for EFX, maximum anchors are sufficient but not necessary for
maximum-singleton dominance, and A-LD is sufficient but not necessary for
extending a precoloring.  These are certificate layers, not structural
assumptions imposed on every valuation instance.
\end{remark}

\section{Recognition of maximum-anchor certificates}\label{sec:recognition}

For active agent \(i\), let
\[
T_i=\{g:v_i(g)=M_i>0\}
\]
be its maximum-good set.  The recognition problem asks whether there is an
injective map \(h_i\in T_i\) satisfying A-LD.

For complexity, we use the combinatorial input
\[
I=(N^+,G^+,(S_g)_{g\in G^+},(T_i)_{i\in N^+}),
\]
where every \(S_g\) and \(T_i\) is nonempty,
\(T_i\subseteq\{g:i\in S_g\}\), and \(C_2\) is derived from the supports.
Let \(|I|\) denote its binary encoding length.  This input is computable in
polynomial bit time from binary-encoded rational valuations.  The
\(O(nm^2)\) bounds in the coloring construction count combinatorial support
comparisons after the supports and maximum sets are available; the complete
rational-input algorithm remains polynomial in \(|I|\).

For fixed \(h\), write
\[
\ell_h(g)=|L_H(g)|.
\]
Repeatedly delete a residual vertex with current degree strictly below
\(\ell_h(g)\).  The terminal set is denoted \(K(h)\).

This is the standard monotone vertex-threshold peeling operation, also viewed
as a heterogeneous core.  For the orientation formulation, first reject if a
residual vertex has \(\ell_h(g)=0\).  On each connected component \(R\) of the
residual conflict graph, the normalization
\(\bar f_h(g)=\max\{0,d_R(g)-\ell_h(g)+1\}\) gives the complementary one-sided
degree-constrained acyclic-orientation criterion known in graph theory
\cite{BaxterEtAl2011,KiralyPalvolgyi2018}.  We record the exact specialization
needed for the support-derived anchor verifier.

\begin{proposition}[Exact fixed-anchor core]\label{prop:core}
The set \(K(h)\) is independent of the deletion order and is the unique maximum
residual set \(K\) satisfying
\[
d_{C_2[K]}(g)\ge\ell_h(g)\qquad(g\in K).
\]
Thus \(h\) satisfies A-LD if and only if \(K(h)=\varnothing\).
\end{proposition}
\begin{proof}
Any set satisfying the displayed lower bounds survives every deletion process:
before its first deleted vertex, all its vertices are still present, so none is
eligible.  Conversely, the terminal set itself satisfies the lower bounds.
It therefore contains every such set and is unique.
\end{proof}

\begin{theorem}[Fixed agents and tie-surplus]\label{thm:recognition}
Let \(a=|N^+|\) and
\[
p=\sum_{i\in N^+}(|T_i|-1).
\]
Maximum-anchor A-LD recognition and search can be solved in
\[
O\!\left(m^a\poly(|I|)\right)
\quad\text{and}\quad
O\!\left(2^p\poly(|I|)\right)
\]
time.  Hence it is polynomial for every fixed number of agents, XP in \(a\),
and FPT in \(p\).
\end{theorem}
\begin{proof}
Enumerate choices from the \(T_i\), discard noninjective maps, and apply
Proposition~\ref{prop:core}.  The first bound follows from
\(\prod_i|T_i|\le m^a\).  For the second, use
\(s\le2^{s-1}\) for every integer \(s\ge1\):
\[
\prod_i|T_i|\le2^{\sum_i(|T_i|-1)}=2^p.\qedhere
\]
\end{proof}

\begin{corollary}
If every active agent has a unique maximum good, recognition first checks
whether these goods are pairwise distinct; a collision gives NO, and otherwise
one fixed-anchor core test suffices.  More generally, if the
agent--maximum-good bipartite graph has a unique saturating matching,
recognition reduces to the core test for that matching.
\end{corollary}

\begin{theorem}[At most two active agents]\label{thm:two-agents}
Suppose \(a=|N^+|\le2\).  For \(a=2\), let \(u\) be the number of goods with
support \(N^+\).  A maximum-anchor A-LD certificate exists if and only if the
agent--maximum-good bipartite graph has a saturating matching and \(u\le2\).
For \(a\le1\), existence of the required top choice is sufficient.
\end{theorem}
\begin{proof}
For \(a\le1\), \(C_2\) has no edges.  For \(a=2\), the \(u\) universal goods
form a clique and all support-one goods are isolated.  Let a top matching use
\(k\) universal anchors.  The residual universal clique has \(r=u-k\) vertices,
each with residual-list size \(2-k\), and it peels exactly when
\[
r-1<2-k,
\]
equivalently \(u\le2\).  Isolated support-one residual goods always peel.
\end{proof}

\begin{definition}[Color-safe top edge]
A candidate edge \((i,h)\), with \(h\in T_i\), is color-safe if for every
\(g\ne h\),
\[
i\in S_g\quad\Longrightarrow\quad gh\notin E(C_2).
\]
\end{definition}

\begin{theorem}[Safe-anchor matching]\label{thm:safe-anchor}
If the full support instance is list-degenerate and the bipartite graph
containing only color-safe top edges has a matching saturating \(N^+\), then
that matching is a maximum-anchor A-LD certificate.
\end{theorem}
\begin{proof}
For a residual good \(g\), color-safety of the chosen edge for every
\(i\in S_g\) implies that no supported color is deleted, so
\(L_H(g)=S_g\).  Every residual induced subgraph is also an induced subgraph of
the full instance and therefore contains a vertex of degree below
\(|S_g|=|L_H(g)|\).
\end{proof}

\begin{example}[A top matching need not pass A-LD]\label{ex:matching-not-enough}
Let
\[
S_a=\{1,2,3\},\quad S_b=S_c=\{1,2\},\quad
S_p=\{2\},\quad S_q=\{3\},
\]
and let \(T_1=\{a\},T_2=\{p\},T_3=\{q\}\).  The unique top matching leaves
adjacent residual goods \(b,c\) with lists \(\{2\},\{2\}\), so its generalized
core is nonempty, even though the unanchored support instance is
list-degenerate.
\end{example}

\begin{remark}
The general recognition problem is in nondeterministic polynomial time (NP).
We do not determine whether it is
polynomial-time solvable or NP-hard.  Example~\ref{ex:matching-not-enough}
shows that a top-good matching alone is not sufficient.
\end{remark}

\section{A fixed-support robustness witness: a restricted
\texorpdfstring{\(4\times10\) \(\EFXz\)}{4 x 10 EFX0} cone}

This construction is a robustness witness on one support face, not a proposed
natural valuation domain.  Its pair multiplicities make the two exceptional
owners and an RI completion transparent, while strict maximum inequalities
make the same anchor certificate valid on a relatively open set of valuations.

Let the ten goods be
\[
(x,y,a_{23},b_{23},a_{14},b_{14},a_{34},b_{34},p_{13},p_{24})
\]
with supports
\[
\begin{aligned}
S_x&=\{1,2,3\},& S_y&=\{1,2,4\},\\
S_{a_{23}}=S_{b_{23}}&=\{2,3\},&
S_{a_{14}}=S_{b_{14}}&=\{1,4\},\\
S_{a_{34}}=S_{b_{34}}&=\{3,4\},&
S_{p_{13}}&=\{1,3\},\quad S_{p_{24}}=\{2,4\}.
\end{aligned}
\tag{M-S}
\]
The exact unary pruning leaves \(C_x=\{1\}\) and \(C_y=\{2\}\).
One RI owner map is
\[
\begin{aligned}
x&\mapsto1,& y&\mapsto2,&
a_{23}&\mapsto2,&b_{23}&\mapsto3,\\
a_{14}&\mapsto1,&b_{14}&\mapsto4,&
a_{34}&\mapsto3,&b_{34}&\mapsto4,\\
p_{13}&\mapsto3,&p_{24}&\mapsto4.
\end{aligned}
\tag{M-O}
\]
It uses each of the twelve ordered observer--owner resources exactly once.

\begin{theorem}[The \(M\)-cone]\label{thm:mcone}
On the fixed support face \((M\text{-}S)\), impose
\[
\begin{aligned}
v_1(x)&>v_1(g) &&(g\ne x),&
v_2(y)&>v_2(g) &&(g\ne y),\\
v_3(b_{23})&>v_3(g) &&(g\ne b_{23}),&
v_4(b_{14})&>v_4(g) &&(g\ne b_{14}),
\end{aligned}
\]
where comparisons only range over the positive support of each row.
Every valuation in this cone has a complete \(\EFXz\) allocation.  The cone is
nonempty and full-dimensional relative to the \(22\)-dimensional support face.
\end{theorem}
\begin{proof}
Use anchors \(h_1=x,h_2=y,h_3=b_{23},h_4=b_{14}\).
The residual lists are
\[
\begin{array}{c|c}
a_{23}&\{2\}\\
a_{14}&\{1\}\\
a_{34},b_{34}&\{3,4\}\\
p_{13}&\{3\}\\
p_{24}&\{4\}.
\end{array}
\]
The residual conflict graph has only the edge \(a_{34}b_{34}\), so A-LD holds.
Theorem~\ref{thm:max-anchor} proves \(\EFXz\).

The support face contains
\(3+3+6\cdot2+2\cdot2=22\) positive coordinates.  The displayed strict
homogeneous inequalities define a relatively open cone and do not lower its
dimension.  For example, assigning value \(4\) to each anchor coordinate and
\(1\) to every other positive coordinate gives an interior point.
\end{proof}

\paragraph{Separation from the published comparison predicates.}
The pattern contains two support-three goods and is not a pure multigraph.  A
six-vertex induced subgraph has degrees equal to list sizes, proving failure of
unanchored support-list degeneracy.  The all-padding argument in
Appendix~\ref{app:separation} shows that it also fails the explicit published
high-girth/controlled-multiplicity predicate of Lianeas et al.
\cite{LianeasEtAl2026}.  These are separations from
displayed sufficient hypotheses, not a claim that the cone lies outside every
known EFX class.

\section{An arbitrary-\texorpdfstring{\(n\)}{n} exclusive-or (XOR)
stress-test family}

The second construction is a symmetric stress test for scalability and
ordered-resource saturation.  Its three types of goods have deliberately
different roles, summarized below.

For \(n\ge4\), define \(\mathcal U_n\) on agents \(\{1,\ldots,n\}\):
\begin{itemize}[leftmargin=2em]
\item two goods \(x,y\) have support \(N\);
\item for every \(i=3,\ldots,n\), there is one pair good \(p_{1i}\) with
support \(\{1,i\}\) and one pair good \(p_{2i}\) with support \(\{2,i\}\);
\item for every \(3\le i<j\le n\), there are two goods with support
\(\{i,j\}\);
\item there is no pair good on \(\{1,2\}\).
\end{itemize}

\begin{center}
\begin{tabular}{p{0.22\linewidth}p{0.31\linewidth}p{0.37\linewidth}}
\toprule
Good type & Combinatorial role & Effect in the RI certificate \\
\midrule
Universal goods \(x,y\) & Create two competing high-support owner choices &
Force the two surviving owners to be used in opposite order \\
Cross pairs \(p_{1i},p_{2i}\) & Connect owners \(1,2\) to each
\(i\ge3\) & Supply the reverse ordered resources after the universal-owner
choice \\
Doubled internal pairs & Place two goods on every \(\{i,j\}\subseteq
\{3,\ldots,n\}\) & Prune unwanted universal owners and saturate both internal
directions \\
\bottomrule
\end{tabular}
\end{center}

The number of goods is
\[
m=2+2(n-2)+2\binom{n-2}{2}=(n-1)(n-2)+2.
\]

\begin{theorem}[XOR ownership and saturated RI]\label{thm:xor}
The surviving owner sets of \(x\) and \(y\) are both \(\{1,2\}\).  Their
feasible owner pairs are exactly
\[
(o(x),o(y))=(1,2)\quad\text{or}\quad(2,1).
\]
Either choice extends to an RI owner map that uses every one of the
\(n(n-1)\) ordered observer--owner resources exactly once.
\end{theorem}
\begin{proof}
Candidates \(1,2\) survive because all incident pair multiplicities are at most
one.  Any candidate \(k\ge3\) is incident to a doubled pair
\(\{k,\ell\}\) for some other \(\ell\ge3\), so exact unary pruning deletes it.
The two universal goods cannot share an owner because that would load every
corresponding ordered resource twice.

For \(o(x)=1,o(y)=2\), assign \(p_{1i}\) and \(p_{2i}\) to \(i\), and split the
two goods on every internal pair \(\{i,j\}\) between \(i\) and \(j\).
The universal goods provide directions into owners \(1,2\), the \(p\)-goods
provide the reverse directions, and the doubled internal pairs provide both
directions inside \(\{3,\ldots,n\}\).  These categories partition all ordered
pairs.  The case \(o(x)=2,o(y)=1\) is symmetric.
\end{proof}

\begin{theorem}[Relative-open \(\EFXz\) cone for \(\mathcal U_n\)]
\label{thm:ucone}
Choose anchors
\[
h_1=x,\qquad h_2=y,\qquad h_i=p_{1i}\quad(3\le i\le n).
\]
All valuations on the fixed support face for which these goods are strict
rowwise maxima admit a complete \(\EFXz\) allocation.  This is a nonempty
relatively open cone of dimension
\[
2n^2-4n+4.
\]
\end{theorem}
\begin{proof}
Every residual \(p_{2i}\) has list \(\{i\}\).  Each doubled internal pair keeps
list \(\{i,j\}\) and induces one isolated edge; there are no other residual
conflicts.  Hence A-LD holds, and Theorem~\ref{thm:max-anchor} applies.

The two universal goods contribute \(2n\) positive coordinates.  The other
\((n-1)(n-2)\) goods have support size two, so the support face has
\[
2n+2(n-1)(n-2)=2n^2-4n+4
\]
coordinates.  Assigning value \(4\) to every anchor coordinate and \(1\) to
every other positive coordinate proves nonemptiness.
\end{proof}

The conflict graph has
\[
|E(C_2)|
=1+2(m-2)+\binom{n-2}{2}
=\frac{5n^2-17n+16}{2}
=\Theta(m).
\]
Thus the family has two high-degree overlap hubs and unboundedly many short
overlaps, but it is not asymptotically dense in the sense
\(|E(C_2)|=\Theta(m^2)\).  Its full vertex set witnesses failure of unanchored
list degeneracy, and its incidence graph is 2-connected.  It is not a pure
multigraph.
Note also that \(\mathcal U_4\) has only eight goods;
the four-agent ten-good separator is Theorem~\ref{thm:mcone}, not
\(\mathcal U_4\).

\section{Related work and separation from prior predicates}

\paragraph{General additive and valuation-restricted results.}
Complete all-good EFX is known for three additive agents and any number of
goods \cite{ChaudhuryGargMehlhorn2024}.  For four additive agents, the most
direct small-goods neighbor is the recent theorem for at most nine goods
\cite{AlkassarFouzMehlhorn2026}.  Several valuation restrictions also yield
complete all-good EFX for arbitrary numbers of agents and goods, including a
common ranking of additive goods, personalized bivalued additive rows, at most
three additive valuation types, and at most two item types
\cite{PlautRoughgarden2020,ByrkaEtAl2026,PrakashEtAl2025,GorantlaEtAl2023}.
Broader monotone valuations have a different current landscape: counterexamples
already occur within the submodular class, while the corresponding additive
existence question remains open \cite{AkramiEtAlCounterexample2026}.
These results constrain input axes different from our support-overlap
condition.  Our \(M\)-cone crosses the four-agent goods-count
boundary only under a fixed support pattern and strict maximum conditions; it
is not an unrestricted \(4\times10\) theorem.  Moreover, \(\mathcal U_4\) has
only eight goods and is covered by the small-goods theorem.  The count
\(m=(n-1)(n-2)+2\) exceeds nine only from \(n=5\), when the number of agents has
also changed.

\paragraph{Graphs and multigraphs.}
Graphical EFX allocation admits existence results even when EFX orientations
need not exist \cite{ChristodoulouEtAl2023}.  Subsequent work covers broad
multigraph classes \cite{SgouritsaSotiriou2025,BhaskarPandit2025}, and recent
work gives complete EFX for arbitrary multigraphs under cancelable valuations
\cite{AfshinmehrEtAl2026}.  The latter theorem already implies existence for
every additive pure-pair instance in our framework.  RI is a more restrictive
certificate and algorithmic encoding; the support-derived machinery developed
here is aimed instead at incorporating genuine high-support goods.
On bipartite multigraphs, Afshinmehr et al. also give broad allocation results
and a parameter-based characterization of EFX orientations
\cite{AfshinmehrEtAlJAAMAS2026}.  Those orientation results concern EFX itself;
our pair-capacity and CSP exactness concern the stricter RI completion
certificate after exceptional owners are fixed.

\paragraph{Orientations and parameterized complexity.}
The orientation problem requires goods to be assigned to incident agents and
asks whether the resulting orientation itself is fair.  This differs from
unrestricted allocation and from recognition of our sufficient RI certificate.
The literature gives envy-free-up-to-one-good (EF1) orientation algorithms,
hardness, and parameterized results for EFX orientations on graphs and
multigraphs
\cite{DeligkasEtAl2025,BlazejEtAl2025,KanellopoulosEtAl2025}.  The binary
symmetric orientation algorithm of Bla\v{z}ej et al. also uses a two-state
2-SAT encoding and a tree notion called a core; its states are component roots
and its core is a leaf-reduced tree, not our exceptional-good owner CSP or
fixed-anchor threshold core.  Those results make it especially
important not to reinterpret our fixed-anchor core test or tie-surplus
enumeration as a classification of general EFX orientation or allocation.
At the graph-theoretic level, our fixed-anchor core test is a specialization of
standard heterogeneous-threshold peeling and the one-sided
degree-constrained-acyclic-orientation criterion
\cite{BaxterEtAl2011,KiralyPalvolgyi2018}; the additional problem here is to
select injective maximum anchors whose owner-labelled deletions make that test
succeed.  Degree-constrained acyclic orientations with general allowed-degree
sets form a broader neighboring problem \cite{GarvardtEtAl2023}.

\paragraph{Hypergraphs and approximate relaxations.}
High-girth hypergraph theorems allow general monotone valuations
\cite{LianeasEtAl2026}, a valuation class broader than ours.  Our displayed
families instead allow short Berge overlaps that fail the explicit
high-girth/controlled-multiplicity hypotheses, at the cost of additivity and
maximum-anchor conditions.  Because those models may include endpoints that
value a good at zero, the comparison uses the all-padding convention in
Definition~\ref{def:zero-padding}.  Almost-EFX and approximate-EFX results trade
exact fairness for
broader hypergraph coverage \cite{KakatelisEtAl2026}; they address a neighboring
but different objective.

In the proofs of Lemmas 3.2 and 3.5, Lianeas et al. preserve EFX by using
pairwise edge uniqueness to ensure that an observer sees at most one relevant
edge, together with their Property 3, which makes the observer's current bundle
dominate that unallocated singleton \cite{LianeasEtAl2026}.  This is the same
local one-visible-good plus dominance motif isolated by
Proposition~\ref{prop:ri-dominance}, although their theorem derives sparsity
from high girth and uses a different dynamic construction.  We make no
independent novelty claim for that elementary fairness step; our formulation
uses it to organize a static RI certificate and separates it from the
maximum-anchor machinery developed here.

The first alternative in Barman et al.'s definition of EFL is that a rival
bundle contains at most one good positively valued by the observer
\cite{BarmanEtAl2018}.  This is a sparse-bundle condition inside the EFL
definition, not an RI definition or an equivalent form of our proposition.
Under maximum-singleton
dominance, RI in fact bounds every rival-bundle value by the observer's own
bundle value; Proposition~\ref{prop:ri-dominance} records only the EFX0
consequence needed for the certificate pipeline.  Accordingly, we use that
proposition as an expository bridge and make no independent novelty claim for
it.

\paragraph{Favorite goods, coloring, and terminology.}
Markakis and Santorinaios obtain \(2/3\)-EFX from distinct favorite goods and
exact EFX from pairwise distinct favorite tiers of size
\(\lfloor m/n\rfloor\) \cite{MarkakisSantorinaios2023}.  Our condition uses one
maximum-valued anchor per active agent together with residual support-list
degeneracy, so its hypotheses and construction are different.  Fair allocation
with exogenous item-conflict graphs already connects allocation and coloring in
another model \cite{ChiarelliEtAl2023}.  Our graph is instead derived from support
overlap, its proper list colorings are exactly RI owner maps, and maximum
anchors convert that mechanism to \(\EFXz\).

Definition names also differ between sources.  For example, the common-ranking
and three-agent results quantify over every removed good, whereas Amanatidis
et al. explicitly distinguish positive-good EFX from all-good EFX0
\cite{PlautRoughgarden2020,ChaudhuryGargMehlhorn2024,AmanatidisEtAl2021}.
Every comparison above is made at the level of the displayed inequality.  We do
not identify two results merely because both use the label ``EFX.''

\section{Limitations and open problems}

\begin{enumerate}[leftmargin=2em]
\item RI is sufficient but not necessary for either \(\EFXp\) or \(\EFXz\).
\item Maximum-singleton dominance is a sufficient value bound within the RI
bridge, not a necessary condition for \(\EFXz\).  Owning a maximum anchor is in
turn only one way to obtain that dominance, and A-LD is only one way to extend
the anchor precoloring to RI.
\item The exact capacity CSP characterizes RI completion, not general
\(\EFXz\) existence.
\item We do not resolve the general complexity of finding a maximum-anchor map
satisfying A-LD.  Our primary-source search through 28 August 2026 found no
exact prior classification of this anchor-selection problem; forward-citation
indexing for the most recent sources remains incomplete.
Fixed-\(n\) polynomiality and tie-surplus fixed-parameter tractability do not
settle this question.
\item The \(M\)-cone is full-dimensional only in its \(22\)-dimensional fixed
support face, not in the ambient \(40\)-dimensional valuation space.
\item The family \(\mathcal U_n\) has only two high-support goods and
\(\Theta(m)\), rather than \(\Theta(m^2)\), conflict edges.
\item We do not solve unrestricted \(4\times10\), unrestricted four-agent
all-good \(\EFXz\), or general additive all-good \(\EFXz\).
\item Our literature search found no earlier theorem with the same
maximum-anchor plus residual-list selection problem, but a search cannot
establish absolute priority.  In particular, recent orientation work and its
forward citations remain incomplete sources of evidence.
\end{enumerate}

Natural next steps are a general complexity classification for maximum-anchor
recognition, a wider precoloring-extension condition than A-LD, and families
whose number of high-support goods grows with \(n\).

\section{Conclusion}

Rival injectivity gives an exact support-derived representation as a
support-list coloring problem.  Pair-capacity elimination and the owner CSP
make the exceptional part algorithmically explicit, while maximum-anchor A-LD
certificates construct RI allocations with the dominance needed for complete
all-good \(\EFXz\).  The elementary RI--dominance implication separates the
fairness check from this construction.  Recognition results expose the
algorithmic boundary of the certificate class, and the finite and scalable
witness families show robustness on support patterns beyond the explicit
pure-multigraph and high-girth/controlled-multiplicity hypotheses compared
here; they do not solve the unrestricted problem.  The remaining challenge is
to classify recognition or find wider extension conditions without weakening
the all-good \(\EFXz\) guarantee.

\appendix

\section{Separation details}\label{app:separation}

For this appendix, the \emph{positive-support incidence graph} is the bipartite
graph on agents and positive-support goods with edge \(ig\) when \(i\in S_g\).
The auxiliary \emph{Block--Hall predicate} requires every nontrivial block of
this incidence graph to have a matching saturating all of its good vertices.
We include it only as a secondary benchmark for the two constructions; it is
not a frontier condition from the cited literature.

\subsection{The \texorpdfstring{\(4\times10\)}{4 x 10} support}

In the six-vertex induced subgraph
\[
\{x,y,a_{23},b_{23},a_{14},b_{14}\},
\]
the two exceptional vertices have degree \(3\), equal to their support-list
size, and the four pair vertices have degree \(2\), again equal to their list
size.  Hence unanchored list degeneracy fails.

The incidence graph has an ear decomposition starting from
\[
1-x-2-y-1,
\]
then adding the ears \(x-3-a_{23}-2\) and \(y-4-a_{14}-1\), followed by
length-two ears for the remaining pair goods.  It is therefore 2-connected and
forms a single block.  Ten good vertices cannot be saturated by four agent
vertices, so the Block--Hall predicate fails.

For zero-endpoint padding, \(x\) and both support-\(\{2,3\}\) goods must receive
the same patch \(P\) to avoid a Berge 2-cycle; similarly \(y\) and both
support-\(\{1,4\}\) goods must receive the same patch \(Q\).  Since \(P,Q\)
share agents \(1,2\), avoiding another Berge 2-cycle forces \(P=Q=N\), with
multiplicity at least six.  This violates the published distinguished-vertex
bound \(6\le |P|-2=2\).  Repeated patches also rule out the simple-hypergraph
branch.

\subsection{The universal family}

The whole conflict graph witnesses failure of unanchored list degeneracy.
Every \(p_{1i},p_{2i}\) has
degree \(2=|S_g|\); every internal doubled-pair good has degree
\(3>|S_g|=2\); and each universal good has degree \(m-1\ge n=|S_g|\).

The incidence graph is 2-connected: after deletion of either universal good,
the other connects all agent vertices, and deletion of any other single vertex
does not disconnect the graph.  Since this is one block and \(m>n\), the
Block--Hall saturation condition fails.

Under zero-endpoint padding, the two universal goods already have patch \(N\).
Any padded pair patch intersects \(N\) in at least two agents; avoiding a Berge
2-cycle therefore forces it also to be \(N\).  All \(m\) goods then share the
same patch, while the cited multiplicity condition would require
\(m\le n-2\), a contradiction.  The underlying simple hypergraph after full
padding may itself have infinite girth; the failure is the parallel
multiplicity bound, not the underlying girth.

\section{Finite verification and provenance}

The general theorems above have symbolic proofs and do not depend on
computation.  The accompanying verification scripts check finite witness
arithmetic:
\begin{itemize}[leftmargin=2em]
\item the \(4\times10\) capacity/cone certificate checks all nine fixed
exceptional-owner choices, the stated RI allocation, A-LD, \(\EFXz\) inequalities,
incidence connectivity, and all \(262{,}144\) zero-endpoint paddings;
\item the \(\mathcal U_n\) script checks identities and explicit allocations for
\(n=4,\ldots,8\), only as finite sanity checks; the all-\(n\) statements rest on
Theorems~\ref{thm:xor} and~\ref{thm:ucone}.
\end{itemize}
No finite no-hit is used as a proof of an all-\(n\) statement.

\section*{Human direction, AI contribution, and responsibility}

\paragraph{Human direction.}
The author selected the EFX research problem and its initial finite target,
supplied the formal task specification and research constraints, and directed
the project through explicit choices of positive and negative research
branches, verification standards, computational-resource limits, stopping and
continuation decisions, and publication goals.  After the initial unrestricted
counterexample search, the author redirected the project toward a publishable
restricted-support theorem, selected RI/capacity/maximum-anchor and recognition
as the contribution center, and retained unrestricted \(4\times10\) EFX as an
open branch.  The author chose the final
contribution package; required conservative separation of symbolic theorems,
finite evidence, conjectures, and unresolved priority; supplied and adjudicated
human-reader feedback; and retains final authority and responsibility for the
manuscript and its release.

\paragraph{AI contribution.}
OpenAI ChatGPT and Codex agent workflows were used during August 2026.  The
agent workflow generated and developed the paper's main structural route:
rival-injective ownership and support-list coloring, pair-capacity elimination
and the exceptional-owner CSP/2-SAT formulation, maximum-anchor residual-A-LD
certificates, recognition subcases, and the displayed \(4\times10\) cone and
the family \(\mathcal U_n\).  It also drafted and revised proof arguments;
implemented exact checkers and finite-certificate scripts; conducted structured
searches and counterexample attempts; used separate agent instances for proof
reconstruction and adversarial review; audited primary-source literature; and
prepared the manuscript, \LaTeX, and typesetting.

\paragraph{Verification and responsibility.}
AI outputs were treated as candidate research material rather than as
mathematical or bibliographic authority.  Symbolic claims retained in this
paper were checked through separate agent reconstruction and/or scoped
adversarial review, according to the claim;
finite computations support only the scoped witness and consistency checks
identified above and do not replace the general proofs; and candidate citations
and priority comparisons were checked against primary sources.  The author
selected the final content and takes responsibility for the paper's accuracy,
originality, and integrity.  At the time of this pre-submission version,
independent human domain-expert review is being sought and should not be inferred
from the agent reviews described here.


\begin{thebibliography}{99}

\bibitem{BarmanEtAl2018}
S.~Barman, A.~Biswas, S.~K.~Krishnamurthy, and Y.~Narahari.
\newblock Groupwise maximin fair allocation of indivisible goods.
\newblock In \emph{Proceedings of the Thirty-Second AAAI Conference on
Artificial Intelligence}, pages 917--924, 2018.
\newblock \url{https://doi.org/10.1609/aaai.v32i1.11463}.

\bibitem{BaxterEtAl2011}
G.~J.~Baxter, S.~N.~Dorogovtsev, A.~V.~Goltsev, and J.~F.~F.~Mendes.
\newblock Heterogeneous k-core versus bootstrap percolation on complex
networks.
\newblock \emph{Physical Review E}, 83:051134, 2011.
\newblock \url{https://doi.org/10.1103/PhysRevE.83.051134}.

\bibitem{KiralyPalvolgyi2018}
Z.~Kir\'aly and D.~P\'alv\"olgyi.
\newblock Acyclic orientations with degree constraints.
\newblock arXiv:1806.03426, 2018.
\newblock \url{https://arxiv.org/abs/1806.03426}.

\bibitem{GarvardtEtAl2023}
J.~Garvardt, M.~Renken, J.~Schestag, and M.~Weller.
\newblock Finding degree-constrained acyclic orientations.
\newblock In \emph{IPEC}, LIPIcs 285, 19:1--19:14, 2023.
\newblock \url{https://doi.org/10.4230/LIPIcs.IPEC.2023.19}.

\bibitem{ChaudhuryGargMehlhorn2024}
B.~R.~Chaudhury, J.~Garg, and K.~Mehlhorn.
\newblock EFX exists for three agents.
\newblock \emph{Journal of the ACM}, 71(1), Article 4, 2024.
\newblock \url{https://doi.org/10.1145/3616009}.

\bibitem{PlautRoughgarden2020}
B.~Plaut and T.~Roughgarden.
\newblock Almost envy-freeness with general valuations.
\newblock \emph{SIAM Journal on Discrete Mathematics}, 34(2):1039--1068,
2020.
\newblock \url{https://doi.org/10.1137/19M124397X}.

\bibitem{ByrkaEtAl2026}
J.~Byrka, F.~Malinka, and T.~Ponitka.
\newblock Probing EFX via PMMS: (Non-)existence results in discrete fair
division.
\newblock In \emph{AAAI}, 40(20):16735--16742, 2026.
\newblock \url{https://doi.org/10.1609/aaai.v40i20.38716}.

\bibitem{PrakashEtAl2025}
V.~Prakash H.V., P.~Ghosal, P.~Nimbhorkar, and N.~Varma.
\newblock EFX exists for three types of agents.
\newblock In \emph{Proceedings of the 26th ACM Conference on Economics and
Computation}, pages 101--128, 2025.
\newblock \url{https://doi.org/10.1145/3736252.3742509}.

\bibitem{GorantlaEtAl2023}
P.~Gorantla, K.~Marwaha, and S.~Velusamy.
\newblock Fair allocation of a multiset of indivisible items.
\newblock In \emph{Proceedings of the 2023 ACM--SIAM Symposium on Discrete
Algorithms}, pages 304--331, 2023.
\newblock \url{https://doi.org/10.1137/1.9781611977554.ch13}.

\bibitem{AmanatidisEtAl2021}
G.~Amanatidis, G.~Birmpas, A.~Filos-Ratsikas, A.~Hollender, and
A.~A.~Voudouris.
\newblock Maximum Nash welfare and other stories about EFX.
\newblock \emph{Theoretical Computer Science}, 863:69--85, 2021.
\newblock \url{https://doi.org/10.1016/j.tcs.2021.02.020}.

\bibitem{ChristodoulouEtAl2023}
G.~Christodoulou, A.~Fiat, E.~Koutsoupias, and A.~Sgouritsa.
\newblock Fair allocation in graphs.
\newblock In \emph{Proceedings of the 24th ACM Conference on Economics and
Computation}, pages 473--488, 2023.
\newblock \url{https://doi.org/10.1145/3580507.3597764}.

\bibitem{BhaskarPandit2025}
U.~Bhaskar and Y.~Pandit.
\newblock Extending EFX allocations to further multi-graph classes.
\newblock In \emph{FSTTCS}, LIPIcs 360, 15:1--15:18, 2025.
\newblock \url{https://doi.org/10.4230/LIPIcs.FSTTCS.2025.15}.

\bibitem{AfshinmehrEtAl2026}
M.~Afshinmehr, A.~Ashuri, P.~Mahmoudkhan, K.~Mehlhorn, and
A.~M.~Shahrezaei.
\newblock EFX allocations exist on multi-graphs.
\newblock arXiv:2606.18665, 2026.
\newblock \url{https://arxiv.org/abs/2606.18665}.

\bibitem{LianeasEtAl2026}
T.~Lianeas, A.~Sgouritsa, and M.~M.~Sotiriou.
\newblock EFX allocation in (multi)hypergraphs.
\newblock In \emph{AAAI}, 40(20):17102--17110, 2026.
\newblock \url{https://doi.org/10.1609/aaai.v40i20.38759}.

\bibitem{KakatelisEtAl2026}
I.~Kakatelis, T.~Lianeas, A.~Sgouritsa, and M.~M.~Sotiriou.
\newblock Almost EFX in hypergraphs.
\newblock arXiv:2606.26948, 2026.
\newblock \url{https://arxiv.org/abs/2606.26948}.

\bibitem{AlkassarFouzMehlhorn2026}
E.~Alkassar, M.~Fouz, and K.~Mehlhorn.
\newblock Complete EFX allocations exist for four additive agents and up to
nine goods.
\newblock arXiv:2608.08590, 2026.
\newblock \url{https://arxiv.org/abs/2608.08590}.

\bibitem{MarkakisSantorinaios2023}
E.~Markakis and C.~Santorinaios.
\newblock Improved EFX approximation guarantees under ordinal-based
assumptions.
\newblock In \emph{AAMAS}, pages 591--599, 2023.
\newblock \url{https://www.ifaamas.org/Proceedings/aamas2023/pdfs/p591.pdf}.

\bibitem{SgouritsaSotiriou2025}
A.~Sgouritsa and M.~M.~Sotiriou.
\newblock On the existence of EFX allocations in multigraphs.
\newblock In \emph{AAMAS}, pages 2735--2737, 2025.
\newblock \url{https://doi.org/10.5555/3709347.3743995}.

\bibitem{DeligkasEtAl2025}
A.~Deligkas, E.~Eiben, T.-L.~Goldsmith, and V.~Korchemna.
\newblock EF1 and EFX orientations.
\newblock In \emph{IJCAI}, pages 56--63, 2025.
\newblock \url{https://doi.org/10.24963/ijcai.2025/7}.

\bibitem{KanellopoulosEtAl2025}
S.~Kanellopoulos, E.~Nemery, C.~Pergaminelis, M.~M.~Sotiriou, and
M.~Vasilakis.
\newblock EF(X) orientations: A parameterized complexity perspective.
\newblock arXiv:2512.25033, 2025.
\newblock \url{https://arxiv.org/abs/2512.25033}.

\bibitem{LohPagh2014}
P.-S.~Loh and R.~Pagh.
\newblock Thresholds for extreme orientability.
\newblock \emph{Algorithmica}, 69(3):522--539, 2014.
\newblock \url{https://doi.org/10.1007/s00453-013-9749-4}.

\bibitem{ChiarelliEtAl2023}
N.~Chiarelli, M.~Krnc, M.~Milani\v{c}, U.~Pferschy, N.~Piva\v{c}, and
J.~Schauer.
\newblock Fair allocation of indivisible items with conflict graphs.
\newblock \emph{Algorithmica}, 85:1459--1489, 2023.
\newblock \url{https://doi.org/10.1007/s00453-022-01079-8}.

\bibitem{AkramiEtAlCounterexample2026}
H.~Akrami, A.~Mayorov, K.~Mehlhorn, S.~Srinivas, and C.~Weidenbach.
\newblock A counterexample to EFX: \(n\ge 3\) agents, \(m\ge n+5\) items,
submodular valuations via SAT-solving.
\newblock arXiv:2604.18216v3, 2026.
\newblock \url{https://arxiv.org/abs/2604.18216}.

\bibitem{AfshinmehrEtAlJAAMAS2026}
M.~Afshinmehr, A.~Danaei, M.~Kazemi, K.~Mehlhorn, and N.~Rathi.
\newblock EFX allocations and orientations on bipartite multi-graphs: A
complete picture.
\newblock \emph{Autonomous Agents and Multi-Agent Systems}, 40:32, 2026.
\newblock \url{https://doi.org/10.1007/s10458-026-09754-8}.

\bibitem{BlazejEtAl2025}
V.~Bla\v{z}ej, S.~Gupta, M.~S.~Ramanujan, and P.~Strulo.
\newblock Tractable graph structures in EFX orientation.
\newblock In \emph{Algorithmic Game Theory (SAGT)}, LNCS 15953, pages
175--190, 2025.
\newblock \url{https://doi.org/10.1007/978-3-032-03639-1_10}.

\end{thebibliography}
\end{document}